\documentclass[final,leqno]{siamltex}

\usepackage[named]{algo}
\usepackage{multirow}

\title{The acceptance-rejection method for low-discrepancy sequences}

\author{Nguyet Nguyen\thanks{Department of Mathematics, Florida State University, Tallahassee FL 32306-4510 ({\tt nnguyen@math.fsu.edu}).}
        \and Giray \"{O}kten\thanks{Department of Mathematics, Florida State University, Tallahassee FL 32306-4510 ({\tt okten@math.fsu.edu}).}}

\begin{document}

\maketitle

\begin{abstract}
Generation of pseudorandom numbers from different probability distributions has been studied extensively in the Monte Carlo simulation literature. Two standard generation techniques are the acceptance-rejection and inverse transformation methods. An alternative approach to Monte Carlo simulation is the quasi-Monte Carlo method, which uses low-discrepancy sequences, instead of pseudorandom numbers, in simulation. Low-discrepancy sequences from different distributions can be obtained by the inverse transformation method, just like for pseudorandom numbers. In this paper, we will present an acceptance-rejection algorithm for low-discrepancy sequences. We will prove a convergence result, and present error bounds. We will then use this acceptance-rejection algorithm to develop quasi-Monte Carlo versions of some well known algorithms to generate beta and gamma distributions, and investigate the efficiency of these algorithms numerically. We will also consider the simulation of the variance gamma model, a model used in computational finance, where the generation of these probability distributions are needed. Our results show that the acceptance-rejection technique can result in significant improvements in computing time over the inverse transformation method in the context of low-discrepancy sequences.
\end{abstract}



\pagestyle{myheadings}
\thispagestyle{plain}
\markboth{NGUYET NGUYEN AND GIRAY \"{O}KTEN}{ACCEPTANCE-REJECTION METHOD}

\section{Introduction}\label{intro}

The Monte Carlo simulation is a popular numerical method across sciences, engineering, statistics, and computational mathematics. In simple terms, the method involves solving a problem by simulating the underlying model using pseudorandom numbers, and then estimates the quantity of interest as a result of the simulation. Simulating the model involves generating pseudorandom numbers from various probability distributions used in the model. There is an extensive literature on algorithms that transform pseudorandom numbers from the uniform distribution to pseudorandom numbers from the target distribution (see, for example, Devroye \cite{devroye}, Fishman \cite{fishman}). Many of these algorithms are based on two main methods and their combinations: the inverse transformation method and the acceptance-rejection method. The latter is used especially when the inverse transformation method is computationally expensive.

An alternative numerical tool to the Monte Carlo (MC) method is the quasi-Monte Carlo (QMC) method. It is easier to describe these methods in the context of numerical integration. Both methods estimate the expected value of a random variable $X$, $E[X]$, using sample means $\frac{1}{N}\sum_{i=1}^{N}X_i$ where $X_1$, ..., $X_N$ are i.i.d. random variables from the distribution of $X$ in MC, and a low-discrepancy sequence (or, a QMC sequence) from the cumulative distribution function (CDF) $F(x)$ of $X$. The definition of low-discrepancy sequences and a comprehensive treatment of its theory can be found in Niederreiter \cite{Nied}. One reason QMC has become popular in some fields such as computational finance is its faster rate of convergence. Theoretical convergence rate of QMC is $O(N^{-1}(\log N)^s)$, where $s$ is the dimension of the integral in the computation of the expectation. This deterministic rate of convergence is asymptotically better than the probabilistic Monte Carlo rate of $O(N^{-0.5})$. However, in many applications, researchers have observed rates close to $O(N^{-1})$ for QMC. We will not discuss the reasons for this better than theoretical rate of convergence which involve concepts like effective dimension and decreasing importance of variables (\cite{caflisch2}).

How do we generate a QMC sequence from a distribution $F(x)$? The process is somewhat similar to MC. One starts with a QMC sequence from the uniform distribution on $(0,1)^s$ and then applies a transformation method to the sequence in order to obtain a sequence from the target distribution. Currently, the only main transformation method used for QMC is the inverse transformation method (the Box-Muller method is also applicable in QMC (\cite{okten goncu}), but its scope is smaller.) The acceptance-rejection method is usually avoided in QMC, though ``smoothed" versions of it were introduced by Moskowitz \& Caflisch \cite{caflisch} and Wang \cite{wang}. The reasons for this avoidance has to do with some theoretical difficulties that involve the inapplicability of Koksma-Hlawka type inequalities to indicator functions with infinite variation.

If the inverse transformation method is computationally expensive for a particular distribution, then its application to a QMC sequence can make the overall QMC simulation too expensive to provide any advantages over the MC simulation. An example of costly inverse transformation algorithm appears in the simulation of a stochastic process known as the variance gamma model by QMC. Avramidis et. al. \cite{lecuyer} comment on the additional cost of computing inverse of beta, gamma, and normal distributions, which are needed in the generation of the variance gamma model, and suggest that this additional cost needs to be considered while assessing the efficiency of different estimators.

In this paper, we present a QMC version of the acceptance-rejection method, prove a convergence result, and develop error bounds. We present QMC algorithms based on acceptance-rejection for the beta and gamma distributions. We illustrate the advantages of these algorithms, and their application to the variance gamma model, numerically. The availability of acceptance-rejection as a transformation method for QMC significantly broadens its scope.

\section{The acceptance-rejection algorithm}

The acceptance-rejection method is one of the standard methods used for generating distributions. Assume we want to generate from the density $f(x)$, and there is another density $g(x)$ (with CDF $G(x)$) we know how to sample from, say, by using the inverse transformation method. Assume the density functions $f(x), g(x)$ have the same domain, $(a,b)$, and there exists a finite constant $C=\sup_{x\in(a,b)}{f(x)/g(x)}$. Let  $h(x)=f(x)/Cg(x)$. The Monte Carlo acceptance-rejection algorithm is:\newline

\begin{algorithm}{Algorithm 1: Acceptance-Rejection algorithm to generate pseudorandom numbers from the density $f(x)$.}

\begin{enumerate}

\item Generate pseudorandom numbers $u$, $v$ from the uniform distribution on $(0,1)$

\item Generate $X$ from $g(x)$ by $X=G^{-1}(u)$

\item If $v\leq h(X)$ accept $X$; Otherwise reject $X$
\item Repeat Steps 1 to 3, until the necessary number of points have been accepted.
\end{enumerate}
\end{algorithm}

Acceptance-rejection is usually avoided in QMC because it involves integration of a characteristic function: this is the step that corresponds to accepting a candidate by a certain probability. Since characteristic functions can have infinite variation in the sense of Hardy and Krause, and since the celebrated Koksma-Hlawka inequality (\cite{Nied}) links the integration error to the variation of the integrand, researchers for the most part have stayed away from the acceptance-rejection method with low-discrepancy sequences. Two notable exceptions are Moskowitz and Caflisch \cite{caflisch} and Wang \cite{wang}. In these papers, smoothed versions of acceptance-rejection are introduced. These methods replace the characteristic functions by continuous ones, thereby removing functions with infinite variation. However, these smoothing methods can be very time consuming; if one considers efficiency (time multiplied by error), the smoothing method can be worse than crude MC simulation. We will present such examples in Section \ref{smooth section}. Perhaps for this reason, the smoothing methods have not gained much ground in applications.

For MC, acceptance-rejection is a very powerful tool. There are several specialized algorithms that combine acceptance-rejection with other techniques to obtain fast simulation methods for many distributions used in computing; for a recent reference see Fishman \cite{fishman}. Currently, the QMC method cannot be effectively used in these algorithms, since the smoothing techniques are expensive.

Let $ \{x_1,...,x_N\}$ be numbers obtained from a QMC algorithm that generates the distribution function $F(x)$. How well these numbers approximate $F(x)$ is given by the $F$-star discrepancy of $\{x_1,...,x_N\}$:

$$ D^*_F(x_1,...,x_N)=\sup_{\alpha \in [a,b]} \left\vert \frac{A([a,\alpha);\{x_1,...,x_N\})}{N}-F(\alpha) \right\vert $$
where $(a,b)$ is the support of $F$, and the function $A([a,\alpha);\{x_1,...,x_N\})$ counts how many numbers in $\{x_1,...,x_N\}$ belong to the interval $[a,\alpha)$. If $F$ is the uniform distribution, we simply write $D^*(x_1,...,x_N)$ and call it star discrepancy. Note that $F$-star discrepancy is the Kolmogorov-Smirnov statistic that measures the distance between the empirical and theoretical distribution functions. In our numerical results we will use the Anderson-Darling statistic which is a generalization of the Kolmogorov-Smirnov statistic (see \cite{stephens}). The Anderson-Darling statistic corresponds to the ``weighted" $F$-star discrepancy of a point set. More on the weighted discrepancy and corresponding Koksma-Hlawka type error bounds can be found in Niederreiter \& Tichy \cite{nied tich} and \"{O}kten \cite{okten error red}.

Next we introduce the acceptance-rejection method for low-discrepancy sequences.

\begin{algorithm}{Algorithm 2: QMC Acceptance-Rejection algorithm to generate a sequence whose $F$-star discrepancy converges to zero.}

\begin{enumerate}
\item Generate a low-discrepancy sequence $\omega$ from the uniform distribution on $(0,1)^2$
$$\omega=\{(u_i,v_i)\in (0,1)^2, i=1,2,...\}$$
\item For $i=1,2,...$
\begin{itemize}
\item Generate $X$ from $g(x)$ by $X=G^{-1}(u_i)$

\item If $v_i\leq h(X)$ accept $X$; otherwise reject $X$
\end{itemize}
\item Stop when the necessary number of points have been accepted.
\end{enumerate}
\end{algorithm}

The algorithm starts with a point set in $(0,1)^2$
\[
\omega_{N}=\{(u_i,v_i), i=1,...,N\}
\]
and then applies inversion (Step 2) to obtain the new point set
\[
P=\{(G^{-1}(u_i),v_i),i=1,...,N\}.
\]
Assume $\kappa(N)$ points are accepted at Step 2 of the algorithm. After a renumbering
of the indices, we obtain the set of ``accepted points'' in $(a,b)$:

\begin{equation} \label {accepted points}
Q_{\kappa(N)}=\{G^{-1}(u_1),...,G^{-1}(u_{\kappa(N)})\}.
\end{equation}

The next theorem shows that the accepted points have $F$-star discrepancy that goes to zero with $N$. This result generalizes Theorem 2.4 of Wang \cite{wang} who proves a similar convergence result when the density $g(x)$ is the uniform density on $(0,1)^s$, and $f(x)$ is a density function on $(0,1)^s$.
\begin{theorem}
\label{main}
We have
\begin{equation}
D^*_F(Q_{\kappa(N)})\rightarrow 0~as~N\rightarrow\infty
\end{equation}
where $D^*_F(Q_{\kappa(N)})$ is the $F$-star discrepancy of the point set $Q_{\kappa(N)}$.
\end{theorem}
\begin{proof}
We need to prove that for any $\alpha\in(a,b)$
\begin{equation}
|F_{\kappa(N)}(\alpha)-F(\alpha)|=\left\vert \frac{A([a,\alpha
);Q_{\kappa(N)})}{\kappa(N)}-F(\alpha)\right\vert \rightarrow0,
\end{equation}
where $F_{\kappa(N)}(\alpha)$ is the empirical CDF.
Define the set
\begin{equation}
\label{ealpha}
\begin{array}{ll}
E(\alpha)&=\left\{  (x,y)\in (0,1)^{2}:G^{-1}(x)<\alpha,y\leq h(G^{-1}
(x))\right\}\\
&=\left\{  (x,y)\in (0,1)^2 :x < G(\alpha),y\leq h(G^{-1}
(x))\right\}
\end{array}
\end{equation}
(for simplicity we will assume $G(x)$ is strictly increasing).
Consider a point $G^{-1}(u_i)\in(a,b),i\in\{1,...,N\}.$ This point
belongs to $Q_{\kappa(N)}$ and falls into $[a,\alpha)$ if and only if

\begin{enumerate}
\item $G^{-1}(u_i)<\alpha$,

\item $G^{-1}(u_i)$ is accepted in Step 2, i.e., $(G^{-1}%
(u_i), v_i)\in P$ is such that $v_i\leq h(G^{-1}(u_i))$.
\end{enumerate}
Therefore, $G^{-1}(u_i)\in[a,\alpha),i\in\{1,...,N\},$
if and only if $(u_i,v_i)  \in E(\alpha)$, which implies
\[
A([a,\alpha);Q_{\kappa(N)})=A(E(\alpha);\omega_{N}).
\]
Now, we work on the local discrepancy:
\begin{equation}
\label{local discrep ineq}
\begin{array}{cc}
& \left\vert \frac{A([0,\alpha);Q_{\kappa(N)})}{\kappa(N)}-F(\alpha)\right\vert \\
& =\left\vert \frac{N}{\kappa(N)}\frac{A(E(\alpha);\omega_{N})}{N}-\frac{N}
{\kappa(N)}Vol(E(\alpha))+\frac{N}{\kappa(N)}Vol(E(\alpha))-F(\alpha)\right\vert \\
& \leq\frac{N}{\kappa(N)}\left\vert \frac{A(E(\alpha);\omega_{N})}{N}
-Vol(E(\alpha))\right\vert +\left\vert \frac{N}{\kappa(N)}Vol(E(\alpha
))-F(\alpha)\right\vert.
\end{array}
\end{equation}

Here $Vol(E(\alpha))$ refers to the Lebesgue measure of the set $E(\alpha)$. Note that $\omega_{N}$ is a u.d. mod 1 sequence in $(0,1)^2$, and the boundary of the set $E(\alpha)$ has Lebesgue measure zero since $h(G^{-1}(x))$ is a continuous function on $(0,1)$. Thus, we have:
\begin{equation}
\label{lim}
\left\vert\frac{A(E(\alpha);\omega_{N})}{N}-Vol(E(\alpha))\right\vert\rightarrow0
\end{equation}
as $N\rightarrow\infty$. Substituting $\alpha=b$ in (\ref{local discrep ineq}), we obtain
\[
\frac{A(E(b);\omega_{N})}{N}\rightarrow Vol(E(b)).
\]
Indeed, note that $(u_i,v_i)$ from $\omega_{N}$ belongs to
$E(b)$ if and only if $v_i\leq h(G^{-1}(u_i)),$ which gives us all the
accepted points, i.e., $A(E(b);\omega_{N})=\kappa(N).$ Then, we have%
\begin{equation}
\label{lim vol}
\frac{\kappa(N)}{N}\rightarrow Vol(E(b)).
\end{equation}
Equations (\ref{lim}) and (\ref{lim vol}) imply the first term of the upper bound of inequality (\ref{local discrep ineq}) converges to zero.
To prove that the second term also goes to zero, it suffices to show that
\begin{equation}
\frac{Vol(E(\alpha))}{Vol(E(b))}-F(\alpha) = 0.
\end{equation}
From (\ref{ealpha}) we have
\begin{equation}
Vol(E(\alpha))=\int_0^{G(\alpha)}\int_0^{h(G^{-1}(x)}dydx=\int_0^{G(\alpha)}h(G^{-1}(x))dx.
\end{equation}
Change of variables yields: $u=G^{-1}(x)$, $du=dx/G'(G^{-1}(x))$, and thus
\begin{equation}
Vol(E(\alpha))=\int_a^{\alpha}h(u)G'(u)du=\int_a^{\alpha}h(u)g(u)du=\frac{1}{C}\int_a^{\alpha}f(u)du=\frac{F(\alpha)}{C}.
\end{equation}
Similarly, we have
\begin{equation}
Vol(E(b))=\frac{1}{C}\int_a^b f(u)du=\frac{1}{C},
\end{equation}
since $f$ is the density function on $(a,b)$. This completes the proof.

\end{proof}

Note that Theorem \ref{main} generalizes to the case when $X$ is an $s$-dimensional random vector in a straightforward way. In Algorithm 2, the low-discrepancy sequence $\omega$ would be replaced by an $(s+1)$-dimensional sequence
$$\omega=\{({\bf{u_i}},v_i)\in (0,1)^{s+1}, i=1,2,...\}$$
where ${\bf{u_i}}\in (0,1)^s.$

\section{Error bounds}
The classical QMC error bound is the celebrated Koksma-Hlawka inequality
$$ \left\vert \frac{1}{N}\sum_{n=1}^{N}f(x_{n})-\int_{\mathcal{X}}fd\mu \right\vert \leq V(f)D^*(x_1,...,x_N)$$
where $V(f)$ is the variation of $f$ in the sense of Hardy and Krause (\cite{Nied}). Indicator functions, unless some conditions are satisfied (\cite {owen variation}), have infinite variation and thus Koksma-Hlawka inequality cannot be used to bound their error. This has been the main theoretical obstacle for the use of low-discrepancy sequences in acceptance-rejection algorithms. As a remedy, smoothing methods (\cite{caflisch}, \cite{wang}) were introduced to replace the indicator functions by smooth functions so that Koksma-Hlawka is applicable. In this section we present error bounds that do not require the bounded variation assumption, and allow the analysis of our QMC Acceptance-Rejection algorithm. In the following section, we will compare our algorithm with the smoothing approach numerically.

Consider a general probability space $(\mathcal{X},
\mathcal{B},\mu)$, where $\mathcal{X}$ is an arbitrary nonempty set, $\mathcal{B}$ is
a $\sigma$-algebra of subsets of $\mathcal{X}$, and $\mu$ is a probability measure
defined on $\mathcal{B}$. Let $\mathcal{M}$ be a nonempty subset of $
\mathcal{B}$. For a point set $\mathcal{P}=\{x_{1},\ldots,x_{N}\}$ and $
M\subseteq \mathcal{X},$ define $A(M;\mathcal{P})$ as the number of elements in $
\mathcal{P}$ that belong to $M.$ A point set $\mathcal{P}$ of $N$ elements
of $\mathcal{X}$ is called $(\mathcal{M},\mu)$-uniform if
\begin{equation}
A(M;\mathcal{P})/N=\mu(M)  \label{uniform pt set equation}
\end{equation}
for all $M\epsilon\mathcal{M}$. The definition of $(\mathcal{M},\mu)$-uniform point sets is due to Niederreiter \cite{nied} who developed error bounds when uniform point sets are
used in QMC integration. A useful feature of these bounds is that
they do not require the integrand to have finite variation. We need the following result from G\"{o}nc\"{u} and \"{O}kten \cite{goncuoktencoll}:

\begin{theorem}
\label{collision}
If $f$ is any bounded $\mu$-integrable function on a probability space $(\mathcal{X},
\mathcal{B},\mu )$ and $\mathcal{M=}\{M_{1},...,M_{K}\}$ a partition of $\mathcal{X},$
then for a point set $\mathcal{P}=\{x_{1},...,x_{N}\}$ we have
\begin{equation}
 \left\vert \frac{1}{N}\sum_{n=1}^{N}f(x_{n})-\int_{\mathcal{X}}fd\mu \right\vert \leq
\sum_{j=1}^{K}\mu (M_{j})(G_{j}(f)-g_{j}(f))+ \\
 \sum_{j=1}^{K}\epsilon _{j,N}\max
(|g_{j}(f)|,|G_{j}(f)|)
\end{equation}
where $\epsilon _{j,N}=\left\vert A(M_{j};\mathcal{P})/N-\mu
(M_{j})\right\vert $, $G_{j}(f)=\sup_{t\in M_{j}}f(t)$ and $g_{j}(f)=\inf_{t\in M_{j}}f(t),$
$1\leq j\leq K.$
\end{theorem}

Theorem \ref{collision} provides a general error bound for any point set $\mathcal{P}$. If the point set is an $(\mathcal{M},\mu)$-uniform point set then the second summation on the right hand side becomes zero and the result simplifies to Theorem 2 of Niederreiter \cite{nied}. Setting $f=1_{S}$, the indicator function of the set $S$, in Theorem \ref{collision}, we obtain a simple error bound for indicator functions:
\begin{corollary}
\label{cor}
Under the assumptions of Theorem \ref{collision}, we have
$$
\label{indicator ineq}
\left\vert \frac{A(S;\mathcal{P})}{N}-\mu(S) \right\vert \leq
\sum_{j=1}^{K}\mu (M_{j})(G_{j}(1_{S})-g_{j}(1_{S}))+\epsilon _{j,N}.
$$
\end{corollary}

Now consider Algorithm 2 (QMC Acceptance-Rejection algorithm) where a low-discrepancy sequence is used to generate the point set $Q_{a(N)}$ (see (\ref {accepted points})). We proved that $|A([a,\alpha);Q_{a(N)})/a(N)-F(\alpha)|\rightarrow 0$ as $N\rightarrow \infty$ in Theorem \ref{main}. Corollary \ref{cor} yields an upper bound for the error of convergence. Indeed, let $S=[a,\alpha)$ for an arbitrary $\alpha \in (a,b)$, $\mathcal{X}$ be the domain for the distribution function $F$, and $\mu$ the corresponding measure. We obtain the following bound:

\begin{equation}
\left\vert \frac{A([a,\alpha);\ Q_{a(N)})}{a(N)}-F(\alpha) \right\vert \leq
\sum_{j=1}^{K}\mu (M_{j})(G_{j}(1_{[a,\alpha)})-g_{j}(1_{[a,\alpha)}))+\epsilon _{j,a(N)}
\end{equation}
If the point set $Q_{a(N)}$ happens to be an $(\mathcal{M},\mu)$-uniform point set with respect to the partition, then the term $\epsilon _{j,a(N)}$ vanishes.

Next, we will discuss randomized quasi-Monte Carlo (RQMC) methods and another error bound that addresses the bounded variation hypothesis. Although QMC methods have a faster asymptotic convergence rate than MC, measuring the actual error of a QMC estimate is not easy.  As a remedy, one can use RQMC methods. These methods allow independent simulations via QMC, and the resulting estimates can be analyzed statistically. The RQMC method uses a family of $s$-dimensional low-discrepancy sequences $\beta_\textbf{u}=\{x_1,x_2,...\}$, indexed by the random parameter $\textbf{u}$. Each sequence $\beta_\textbf{u}$ gives rise to the quadrature rule
\begin{equation}\label{Qu}
Q(\beta_\textbf{u})=\frac{1}{N}\sum_{i=1}^N f(x_i).
\end{equation}
Then, $I(f)=\int_{(0,1)^s}f(x)dx$ is estimated by taking the average of $M$ samples
\begin{equation}
I(f)\simeq\frac{1}{M}\sum_{m=1}^M Q(\beta_{\textbf{u}_m}).
\end{equation}
RQMC has three general properties:
\begin{enumerate}
\item $E[Q(\beta_\textbf{u})]=I(f)$
\item $Var(Q(\beta_\textbf{u})=O(N^{-2}(\log N)^{2s})$
\item $|Q(\beta_\textbf{u})-I(f)|\leq V(f)D^*(\beta_\textbf{u})$
\end{enumerate}

Let $\mathcal F$ be the class of real continuous functions defined on $[0,1)^s$ and equipped with Wiener sheet measure $\mu$. Theorem \ref {Hickernell thm}  shows that the mean variance of $Q(\beta_\textbf{u})$ under this measure is $O(N^{-2}(\log N)^{2s})$. Since a function $f(x)$ chosen from the Brownian sheet measure has unbounded variation with probability one, this result provides an alternative error analysis approach to classical Koksma-Hlawka inequality which requires the integrand to be of finite variation. This result was obtained by Wang and Hickernell \cite {wang hickernell} (Theorem 5, page 894) for a particular RQMC method called ``random-start Halton sequences". However, their proof is valid for any RQMC method.

\begin{theorem}\label{Hickernell thm}
The average variance of the estimator, $Q(\beta_\textbf{u})$, taken over function set ${\cal F}$, equipped with the Brownian sheet measure $d\mu$, is:
$$\int_{\mathcal{F}}E[Q(\beta_\textbf{u})-I(f)]^2 d\mu = O(N^{-2}(\log N)^{2s}).$$

\end{theorem}

In our numerical results that follow, we use random-start Halton sequences (\cite{okten}, \cite {wang hickernell}). Theorem \ref{collision} can be used to analyze error for both inverse transformation and acceptance-rejection implementations that we will discuss. Theorem \ref{Hickernell thm} applies only for the inverse transformation implementations, since it is not known whether the accepted points given by the acceptance-rejection algorithm satisfy the discrepancy bound $O(N^{-1}(\log N)^{s})$.

\section{Smoothing} \label{smooth section}

In this section we will compare the QMC Acceptance-Rejection algorithm with the smoothed acceptance-rejection algorithms by Moskowitz \& Caflisch \cite{caflisch}, and Wang \cite{wang}. The algorithms will be compared numerically in terms of efficiency, which is defined as sample variance times computation time. We will use the same numerical examples that were considered in \cite{caflisch} and \cite{wang}.

Consider the problem of estimating the integral $I(f)=\int_{(0,1)^s}f(x)dx$ using the importance function $p(x)$
$$ I(f)=\int_{(0,1)^s}\frac{f(x)}{p(x)}p(x)dx.$$ The MC estimator for $I(f)$ is
$$ \widetilde{I}(f)=\frac{1}{N}\sum_{i=1}^{N}\frac{f(x_i)}{p(x_i)}, x_i\sim p(x).$$
The standard acceptance-rejection algorithm, Algorithm 1, takes the following form for this problem:

\begin{algorithm}{}{Acceptance-Rejection}
\\
\begin{enumerate}
    \item Select $\gamma\geq \sup_{x\in (0,1)^s}p(x)$
		\item Repeat until $N$ points have been accepted:
		\begin{itemize}
	   	\item Sample $x_i\in (0,1)^s$, $y_i\in (0,1)$
			\item If $y_i<\frac{p(x_i)}{\gamma}$, accept $x_i$\\
			      Otherwise, reject $x_i$
			
		\end{itemize}
\end{enumerate}
\end{algorithm}

The smoothed acceptance-rejection method of Moskowitz and Caflisch \cite{caflisch} introduces a weight function $w(x,y)$ such that $$\int_0^1w(x,y)dy=\frac{p(x)}{\gamma}, x\in (0,1)^s, y\in (0,1).$$ The weight function $w(x,y)$ is generated by the following algorithm we call SAR1.

\begin{algorithm}{}{Algorithm 3 (SAR1): Smoothed acceptance-rejection by Moskowitz and Caflisch \cite{caflisch}}\label{A1}
\\
\begin{enumerate}
\item  Select $\gamma\geq \sup_{x\in (0,1)^s}p(x)$, and $0<\sigma<<1$
\item Repeat until weight of accepted points is within one unit of $N$:
\begin{itemize}
\item Sample $x_i\in (0,1)^s$, $y_i\in (0,1)$
\item If $y_i < \frac{p(x_i)}{\gamma}-\frac{1}{2}\sigma$ set $w=1$\\
      Else if $y_i > \frac{p(x_i)}{\gamma}+\frac{1}{2}\sigma$ set $w=0$\\
			Else set $w=\frac{1}{\sigma}(\frac{p(x_i)}{\gamma}+\frac{1}{2}\sigma-y_i)$

\end{itemize}

\end{enumerate}
\end{algorithm}

Wang \cite{wang} extended the SAR1 algorithm by choosing functions $A(x)$, $B(x)$ such that
$$0\leq A(x)<p(x)<B(x)\leq \gamma, ~x\in (0,1)^s, \gamma\geq \sup_{x\in (0,1)^s}p(x),$$
and setting the weight function using the following algorithm (which we call SAR2).

\begin{algorithm}{}{Algorithm 4 (SAR2): Smoothed acceptance-rejection by Wang \cite{wang}}\label{A2}
\\
\begin{enumerate}
\item  Select $\gamma\geq \sup_{x\in (0,1)^s}p(x)$, and functions $A(x)$, $B(x)$ such that
$$0\leq A(x)<p(x)<B(x)\leq \gamma, x\in (0,1)^s$$
\item Repeat until weight of accepted points is within one unit of $N$:
\begin{itemize}
\item Sample $x_i\in (0,1)^s$, $y_i\in (0,1)$
\item If $y_i < \frac{A(x_i)}{\gamma}$ set $w=1$\\
      Else if $y_i\geq \frac{B(x)}{\gamma}$ set $w=0$\\
      Else if $\frac{p(x_i)}{\gamma}\geq y_i > \frac{A(x_i)}{\gamma}$ set  $w=1+\frac{(p(x_i)-B(x_i))(\gamma y_i-A(x_i))}{(B(x_i)-A(x_i))(p(x_i)-A(x_i))}$\\
			Else set $w=\frac{(p(x_i)-A(x_i))(\gamma y_i-B(x_i))}{(B(x_i)-A(x_i))(p(x_i)-B(x_i))}$

\end{itemize}

\end{enumerate}
\end{algorithm}

Now we consider the example used in \cite{caflisch} (Example 3, page 43) and \cite{wang}. The problem is to estimate the integral $I(f)=\int_{(0,1)^s}f(x)dx$, where $s=7$ and
 $$f(x)=\exp(1-(\sin^2(\frac{\pi}{2}x_1)+\sin^2(\frac{\pi}{2}x_2)+\sin^2(\frac{\pi}{2}x_3)))\arcsin(\sin(1)+\frac{x_1+...+x_7}{200}).$$
The importance function is
 $$
p(x)=\frac{1}{C}\exp(1-(\sin^2(\frac{\pi}{2}x_1)+\sin^2(\frac{\pi}{2}x_2)+\sin^2(\frac{\pi}{2}x_3))),
$$
where
$$C=\int_{(0,1)^7}\exp(1-(\sin^2(\frac{\pi}{2}x_1)+\sin^2(\frac{\pi}{2}x_2)+\sin^2(\frac{\pi}{2}x_3)))dx=(\int_0^1\exp(-\sin^2(\frac{\pi}{2}x))dx)^3.$$
Three estimators are used:
\begin{itemize}
\item Crude Monte Carlo (CR):
 $$\frac{1}{N}\sum_{i=1}^N f(x_i),~x_i\sim U((0,1)^s)$$
\item Acceptance-Rejection (AR):
$$\frac{1}{N}\sum_{i=1}^N f(x_i)/p(x_i),~x_i\sim p(x),~x_i \textit{ are accepted points}$$
\item Smoothed Acceptance-Rejection (SAR1 and SAR2): $$\frac{1}{N}\sum_{i=1}^{N^*} w(x_i,y_i)f(x_i)/p(x_i),$$
where $N^*$ is a positive integer such that $\sum_{i=1}^{N^*}w(x_i,y_i)$ is approximately $N$.

\end{itemize}
Table \ref{W1} displays the efficiency of the algorithms. We normalize the efficiency of the algorithms by the efficiency of the crude Monte Carlo algorithm. For example, the efficiency of the Acceptance-Rejection (AR) algorithm, $\textit{Eff}_{AR}$ is computed by
\begin{equation}\label{E1}
\textit{Eff}_{AR}=\frac{\sigma_{CR}^2 \times t_{CR}}{\sigma_{AR}^2 \times t_{AR}},
\end{equation}
where $\sigma_{CR}$ is the sample standard deviation of $M$ estimates obtained using the crude Monte Carlo algorithm, and $t_{CR}$ is the corresponding computing time. Similarly, the parameters $\sigma_{AR}$ and $t_{AR}$ refer to the sample standard deviation and computing time for the Acceptance-Rejection (AR) algorithm.

Although we are primarily interested in how these algorithms compare when they are used with low-discrepancy sequences, for reference, we also report efficiencies when the algorithms are used with pseudorandom numbers. The first part of the table reports the Monte Carlo values (MC) where the pseudorandom sequence Mersenne twister \cite{mersenne} is used, and the second part reports the (randomized) quasi-Monte Carlo (RQMC) values where random-start Halton sequences (\cite{okten}, \cite{wang hickernell}) are used.

In the numerical results, $M=64$, $\sigma=0.2$ in the algorithm SAR1, and $A(x)=1/Ce^2$, $B(x)=e/C$ in the algorithm SAR2. We consider the same sample sizes $N$ as in \cite{caflisch} so that our results can be compared with theirs. Table \ref{W1} reports the sample standard deviation and efficiency (in parenthesis) for each algorithm. Note that in our notation, larger efficiency values suggest the method is better.

Based on the numerical results in Table \ref{W1}, we make the following conclusions. In QMC, the Acceptance-Rejection (AR) algorithm has better efficiency than the smoothed algorithms SAR1 and SAR2, by approximately factors between 2 and 28. A part of the improved efficiency is due to the faster computing time of the AR algorithm. However, the AR algorithm also provides lower standard deviation for all samples. In the case of MC, the AR algorithm has still better efficiency, but with a smaller factor of improvement. 

\begin{table}[h]

\centering
\begin{tabular}{c|cccc|cccc}
\hline
\multirow{2}{*} {N} & \multicolumn{4}{|c|}{MC}& \multicolumn{4}{|c}{QMC}\\
\cline{2-5}\cline{6-9}
 &CR&SAR1&SAR2&AR&CR&SAR1& SAR2& AR\\
\hline
\hline
 256&$3.1e^{-2}$&8.8$e^{-4}$&8.8$e^{-4}$&3.2$e^{-4}$&2.7$e^{-3}$&7.0$e^{-4}$&7.3$e^{-4}$&1.5$e^{-4}$\\

&(1)&(211)&(317)&(3972)&(57)&(266)&(255)&(7233) \\
\hline
1024&1.6$e^{-2}$&  2.7$e^{-4}$&2.3$e^{-4}$&1.4$e^{-4}$&6.2$e^{-4}$  &2.1$e^{-4}$&2.2$e^{-4}$&7.8$e^{-5}$ \\

&(1)  &(652)&(1252)&(5829)& (284)     &(686)&(656) &(6359)\\
\hline
4096&8.8$e^{-3}$&    7.8$e^{-5}$&8.8$e^{-5}$&8.0$e^{-5}$ & 1.6$e^{-4}$  &5.0$e^{-5}$&5.0$e^{-5}$&2.6$e^{-5}$\\
&(1)&   (2595)&(2600)&(5610)& (1430)    &(3872)&(3872)& (20308)\\
\hline
16384&3.9$e^{-3}$&    3.6$e^{-5}$&3.5$e^{-5}$&4.2$e^{-5}$& 4.2$e^{-5}$    &1.3$e^{-5}$&1.2$e^{-5}$ &9.6$e^{-6}$\\

&(1)&      (2492)&(2786)&(3900)&(3900)   &(9900)&(12350) &(25594) \\
\hline
\end{tabular}
\caption{Comparison of acceptance-rejection algorithm AR with its smoothed versions SAR1 and SAR2, in terms of sample standard deviation and efficiency (in parenthesis).} \label{W1}
\end{table}

\section{Applications}

Our main motivation is to develop fast and accurate QMC algorithms for the simulation of a particular L\'{e}vy process known as the variance gamma model (\cite{madan}, \cite{madan2}). This model is used in financial mathematics, and its QMC simulation is expensive due to the inverse transformation method as mentioned in Section \ref{intro}. There are several ways to simulate the variance gamma model (\cite{fu}) and the methods involve generation of normal, beta, and gamma distributions. We will present QMC algorithms based on acceptance-rejection for generating beta and gamma distributions, and numerically compare them with their counterparts based on the inverse transformation method. Then we will present numerical results from the pricing of financial options under the variance gamma model. In all the numerical results, Mersenne twister \cite{mersenne} is used for Monte Carlo, and random-start Halton sequences (\cite{okten}, \cite{wang hickernell}) are used for quasi-Monte Carlo. Clearly, the advantages of the algorithms we present for beta and gamma distributions go beyond the specific application we consider in this paper.

\subsection{Generating beta distribution}

The beta distribution, ${\cal B}(\alpha,\beta)$, has the density function
\begin{equation}
f(x)=\frac{x^{\alpha-1}(1-x)^{\beta-1}}{\textbf{B}(\alpha,\beta)},~0<x<1,
\end{equation}
where $\alpha,\beta>0$ are shape parameters, and $\textbf{B}(\alpha,\beta)$ is the beta function,
\begin{equation}
\textbf{B}(\alpha,\beta)=\int_0^1 t^{\alpha-1}(1-t)^{\beta-1}dt.
\end{equation}

There are different algorithms for the generation of the beta distribution ${\cal B}(\alpha,\beta)$ depending on whether $\min(\alpha,\beta)>1$, $\max(\alpha,\beta)<1$, or neither. The inverse transformation method is especially slow when $\alpha$ or $\beta$ is small. For this reason, we will concentrate on the case $\max(\alpha,\beta)<1$, and the Algorithm AW by Atkinson and Whittaker \cite{fishman}. The algorithm uses a combination of composition, inverse transformation, and acceptance-rejection methods. We next introduce a QMC version of this algorithm.

\begin{algorithm} {Algorithm 7: QMC-AW for generating $X\sim{\cal B}(\alpha,\beta)$ distribution where $\max(\alpha,\beta)<1$.}
\begin{enumerate}
\item Input parameters $\alpha,\beta$
\item Set $t=1/[1+\sqrt{\beta(1-\beta)/\alpha(1-\alpha)]}$,~$p=\beta t/[\beta t+\alpha(1-t)]$
\item Generate a low-discrepancy sequence $\omega$ on $(0,1)^2$
$$\omega=\{(u_i,v_i)\in (0,1)^2, i=1,2,...\}$$
\item For $i=1,2,...$
  \begin{itemize}
      \item Set $Y=-\log u_i$
      \item If $v_i\leq p$
         \begin{itemize}
            \item $X=t(v_i/p)^{1/\alpha}$
            \item If $Y\geq (1-\beta)(t-X)/(1-t)$, accept $X$
            \item If $Y\geq (1-\beta)\log((1-X)/(1-t))$, accept $X$
            \item Otherwise reject $X$
         \end{itemize}
     \item Otherwise
         \begin{itemize}
            \item $X=1-(1-t)[(1-u_i)/(1-p)]^{1/\beta}$
            \item If $Y\geq (1-\alpha)(X/t-1)$, accept $X$
            \item If $Y\geq (1-\alpha)\log(X/t)$, accept $X$
            \item Otherwise reject $X$
         \end{itemize}
    \end{itemize}
\item Stop when the necessary number of points have been accepted.
\end{enumerate}
\end{algorithm}

In Table \ref{beta table} we consider several values for $\alpha$ and $\beta$ that are less than one. \textbf{AR MC} and \textbf{AR QMC} in the table refer to the MC version of Algorithm AW, and its QMC version (Algorithm 7), respectively. The inverse transformation method\footnote{The inverse transformation code we used is a C++ code written by John Burkardt (http://people.sc.fsu.edu/~jburkardt/), and it is based on algorithms by Cran et. al. \cite{cran} and Majumder and Bhattacharjee \cite{majumder}. The performance of the inverse transformation method greatly depends on the choice of tolerance for the method. A large tolerance can result in values that fail the Anderson-Darling goodness-of-fit test. A smaller tolerance increases the computing time. Therefore, in our numerical results, we set tolerances for different range of parameter values small enough so that the results pass the goodness-of-fit test. For $\alpha$, $\beta<1$, we set the tolerance to $10^{-8}$.} with MC and QMC is labeled as \textbf{Inverse MC} and \textbf{Inverse QMC}. We generate $10^5$ numbers from each distribution, and record the computing time and the Anderson-Darling statistic of the sample, in Table \ref{beta table}. 

\begin{table}[h]
\centering
\begin{tabular}{cc|ccccc}
\hline
\multirow{2}{*} {Algorithms} &&Inverse&AR &Inverse&AR\\
 & &MC&MC&QMC&QMC\\
\hline
\hline

\multirow{2}{*}{${\cal B}(0.3,0.3)$}&Time(s)&0.32&0.04&0.32&0.04\\

&$A^2$&4.07e-1&1.35&1.13e-1&8.7e-4\\
\hline

\multirow{2}{*}{${\cal B}(0.3,0.5)$}&Time(s)&0.33& 0.03&0.33&0.03\\

&$A^2$&1.67&7.65e-1&8.64e-2&2.24e-3\\
\hline

\multirow{2}{*}{${\cal B}(0.3,0.7)$}&Time(s)&0.34& 0.03&0.34&0.03\\
&$A^2$&1.65&8.81e-1&1.33e-1&7.5e-4\\

\hline

\multirow{2}{*}{${\cal B}(0.5,0.3)$}&Time(s)&0.33&0.03&0.33&0.03\\

&$A^2$&2.77&8.39e-1&9.22e-2&6.4e-4\\
\hline
\multirow{2}{*}{${\cal B}(0.5,0.5)$}&Time(s)&0.32& 0.03&0.32&0.02\\

&$A^2$&2.34&5.68e-1&2.1e-4&2.56e-3\\
\hline
\multirow{2}{*}{${\cal B}(0.5,0.7)$}&Time(s)&0.33& 0.03&0.34&0.02\\

&$A^2$&6.24e-1&5.47e-1&2.5e-4&5.5e-4\\
\hline
\multirow{2}{*}{${\cal B}(0.7,0.3)$}&Time(s)&0.33& 0.03&0.33&0.03\\

&$A^2$&6.86e-1&7.12e-1&1.35e-1&1.49e-3\\
\hline
\multirow{2}{*}{${\cal B}(0.7,0.5)$}&Time(s)&0.33& 0.03&0.33&0.03\\

&$A^2$&1.92&2.15&2.7e-4&8.9e-4\\
\hline
\multirow{2}{*}{${\cal B}(0.7,0.7)$}&Time(s)&0.33& 0.03&0.35&0.03\\
&$A^2$&8.99e-1&2.06&1.6e-4&5.7e-4\\

\hline

\end{tabular}
\caption{Comparison of inverse and acceptance-rejection algorithms, Algorithm AW (for MC) and Algorithm 7 (QMC-AW), in terms of the computing time and the Anderson-Darling statistic of the sample for the Beta distribution when $N=10^5$ numbers are generated. The percentage points for the $A^2$ statistic at 5\% and 10\% levels are 2.49 and 1.93, respectively.}
\label{beta table}
\end{table}

We make the following observations:
\begin{enumerate}
\item The Acceptance-Rejection algorithm runs about 10 times faster than the inverse transformation algorithm, in both MC and QMC implementations. There is no significant difference in the computing times between MC and QMC, for each algorithm.
\item Inverse MC fails the Anderson-Darling test at the 5\% level for ${\cal B}(0.5,0.3)$. There are several Anderson-Darling values close to the 10\% percentage point $1.93$ for Inverse MC and AR MC methods. Switching to QMC improves the Anderson-Darling values significantly for both inverse and acceptance-rejection, implying better fit of the samples to the theoretical distribution. The Inverse QMC Anderson-Darling values range between $10^{-1}$ and $10^{-4}$. The AR QMC Anderson-Darling values are more stable, and range between $10^{-3}$ and $10^{-4}$.

\end{enumerate}

\subsection{Generating gamma distribution}

The gamma distribution, ${\cal G}(\alpha,\beta)$, has the following property: if $X$ is a random variable from ${\cal G}(\alpha,1)$ then $\beta X$ is the random variable from ${\cal G}(\alpha, \beta)$. Therefore, we only need algorithms to generate random variables from ${\cal G}(\alpha,1)$, which has the density function
\begin{equation}
f(x)=\frac{x^{\alpha-1}e^{-x}}{{\bf\Gamma}(\alpha)},~~~~~\alpha>0,~ x\geq 0.
\end{equation}
Here $\bf\Gamma$ is the gamma function
\begin{equation}
{\bf\Gamma}(z)=\int_0^\infty e^{-t}t^{z-1}dt,
\end{equation}
where $z$ is a complex number with positive real part.

We will consider two algorithms for generating the gamma distribution and present their QMC versions. The algorithms are:
\begin{itemize}
\item Algorithm CH by Cheng \cite{Cheng}, which is applicable when $\alpha>1$,
\item Algorithm GS* by Ahrens-Dieter \cite{fishman}, which is applicable when $\alpha<1$.
\end{itemize}

Next we introduce the QMC versions of these algorithms.

\begin{algorithm}
Algorithm 8: QMC-CH for generating $X\sim{\cal G}(\alpha,1)$ distribution where $\alpha >1$
\begin{enumerate}
\item Input gamma parameter $\alpha$
\item Set $a=(2\alpha-1)^{-1/2},~b=\alpha-\log 4,~c=\alpha+a^{-1}$
\item Generate a low-discrepancy sequence $\omega$ on $(0,1)^2$
$$\omega=\{(u_i,v_i)\in (0,1)^2, i=1,2,...\}$$
\item For $i=1,2,...$
  \begin{itemize}
    \item Set $Y=a\log(\frac{u_i}{1-u_i}),~X=\alpha e^Y,~Z=u_i^2v_i,~R=b+cY-X$
    \item If $R+2.5040774-4.5Z\geq 0$, accept $X$
    \item If $R\geq\log Z$, accept $X$
    \item Otherwise reject $X$
  \end{itemize}
\item Stop when the necessary number of points have been accepted.	
\end{enumerate}
\end{algorithm}

\begin{algorithm}{Algorithm 9: QMC-GS* for generating $X\sim{\cal G}(\alpha,1)$ distribution where $\alpha<1$}
\begin{enumerate}
\item Input gamma parameter $\alpha$
\item Set  $b=(\alpha+e)/e$
\item Generate a low-discrepancy sequence $\omega$ on $(0,1)^3$
$$\omega=\{(u_i,v_i,w_i)\in (0,1)^3, i=1,2,...\}$$
\item For $i=1,2,...$
  \begin{itemize}

         \item Set $Y=bu_i$

         \item If $Y\leq 1$, set $X=Y^{1/\alpha}$
             \begin{itemize}
                \item  Set $W=-\log v_i$ 
                \item  If $W\geq X$ accept $X$
                \item Otherwise reject $X$
             \end{itemize}
         \item Otherwise
             \begin{itemize}
                 \item Set $X=-\log[(b-Y)/\alpha]$
                 \item Set $W=w_i^{1/(\alpha-1)}$
                 \item If $W\leq X$ accept $X$
                 \item Otherwise reject $X$
              \end{itemize}
    \end{itemize}
\item Stop when the necessary number of points have been accepted.
\end{enumerate}
\end{algorithm}

In Tables \ref{gamma1} and \ref{gamma2}, we consider several parameters for $\alpha$ in ${\cal G}(\alpha,1)$. Table \ref{gamma1} has $\alpha$ values greater than one, and Table \ref{gamma2} has values less than one. The parameters are chosen roughly in the range that was observed in the simulation of the variance gamma option pricing problem we will discuss later.

We generate $10^6$ numbers from the corresponding gamma distribution, and compute the execution time of the algorithm and the Anderson-Darling statistic of the sample. Table \ref{gamma1} gives the results for Algorithm CH and Table \ref{gamma2} for Algorithm GS*. The labels \textbf{AR MC} and \textbf{AR QMC} in the tables refer to the MC and QMC versions of Algorithm CH and Algorithm GS*. The inverse transformation method\footnote{The inverse transformation code we used is a C++ code written by John Burkardt, and it is based on two algorithms by Lau \cite{Lau} and Mcleod \cite{Mcleod}. Similar to the beta distribution, the performance of the inverse transformation method greatly depends on the choice for the tolerance parameter. In our numerical results, we set tolerances for different range of values for $\alpha$ small enough so that the results pass the goodness-of-fit test. For example, for $0.25>\alpha \ge 0.20$, we set the tolerance to $10^{-14}$, while we only need to set the tolerance to $10^{-6}$ for $\alpha >1$. The convergence of the inverse transformation method was especially problematic for smaller $\alpha$, for example, when $\alpha<0.1$.}
with MC and QMC is labeled as \textbf{Inverse MC} and \textbf{Inverse QMC}.

\begin{table}[h]
\centering
\begin{tabular}{cc|ccccc}
\hline
\multirow{2}{*} {Algorithms} &&Inverse&AR &Inverse&AR\\
 & &MC&MC&QMC&QMC\\
\hline
\hline

\multirow{2}{*}{${\cal G}(1.6,1)$}&Time(s)&10.62&0.30&10.60& 0.29\\

&$A^2$&2.90e-1&1.09&1.38e-4&8.6e-4\\
\hline

\multirow{2}{*}{${\cal G}(2.0,1)$}&Time(s)&10.92& 0.29&8.21&0.28\\

&$A^2$&8.78e-1&7.52e-1&1.93e-4&1.78e-3\\
\hline

\multirow{2}{*}{${\cal G}(2.4,1)$}&Time(s)&12.32& 0.29&11.74&0.28\\
&$A^2$&9.57e-1&1.83&1.0e-4&2.2e-4\\

\hline

\multirow{2}{*}{${\cal G}(2.8,1)$}&Time(s)&12.41&0.28&13.17&0.27\\

&$A^2$&1.01&5.09e-1&1.1e-4&2.34e-3\\
\hline
\multirow{2}{*}{${\cal G}(3.2,1)$}&Time(s)&12.62& 0.28&12.64&0.27\\

&$A^2$&7.83e-1&5.67e-1&1.3e-4&1.21e-3\\
\hline
\end{tabular}
\caption{Comparison of inverse and acceptance-rejection algorithms, Algorithm CH (for MC) and Algorithm 8 (QMC-CH), in terms of the computing time and the Anderson-Darling statistic of the sample for the Gamma distribution when $N=10^6$ numbers are generated. The percentage points for the $A^2$ statistic at 5\% and 10\% levels are 2.49 and 1.93, respectively.}
\label{gamma1}
\end{table}

We make the following observations based on Table \ref{gamma1}:
\begin{enumerate}
\item The AR QMC algorithm runs faster than the Inverse QMC algorithm by approximately factors between 30 and 48. Interestingly, the AR QMC algorithm is slightly faster than the AR MC algorithm for each case. Similarly, the AR MC is faster than Inverse MC, at about the same factors of speed up.
\item All samples pass the Anderson-Darling test at the 5\% level. Switching to QMC drastically lowers the Anderson-Darling values: Inverse QMC values are around $10^{-4}$, and AR QMC values range between $10^{-3}$ and $10^{-4}$. 
\end{enumerate}

\begin{table}[h]
\centering
\begin{tabular}{cc|ccccc}
\hline
\multirow{2}{*} {Algorithms} &&Inverse&AR &Inverse&AR\\
 & &MC&MC&QMC&QMC\\
\hline
\hline

\multirow{2}{*}{${\cal G}(0.2,1)$}&Time(s)&16.77& 0.25&16.77& 0.24\\

&$A^2$&1.73&5.18e-1&1.30&2.8e-4\\
\hline

\multirow{2}{*}{${\cal G}(0.4,1)$}&Time(s)&13.47&0.33&13.49&0.35\\

&$A^2$&6.76e-1&6.75e-1&4.28e-3&3.5e-4\\
\hline

\multirow{2}{*}{${\cal G}(0.6,1)$}&Time(s)&7.74&0.35&7.73&0.36\\

&$A^2$&1.64&1.01&5.15e-2&6.2e-4\\
\hline

\multirow{2}{*}{${\cal G}(0.8,1)$}&Time(s)&8.07& 0.36&8.08&0.36\\

&$A^2$&5.18e-1&1.29&6.99e-3&3.1e-4\\
\hline
\end{tabular}
\caption{Comparison of inverse and acceptance-rejection algorithms, Algorithm GS* (for MC) and Algorithm 9 (QMC-GS*), in terms of the computing time and the Anderson-Darling statistic of the sample for the Gamma distribution when $N=10^6$ numbers are generated. The percentage points for the $A^2$ statistic at 5\% and 10\% levels are 2.49 and 1.93, respectively.}
\label{gamma2}
\end{table}

We make the following observations based on Table \ref{gamma2}:
\begin{enumerate}
\item The AR QMC algorithm runs faster than the Inverse QMC algorithm by approximately factors between 22 and 70. For smaller $\alpha$ values, the convergence of the inverse transformation method is particularly slow.
\item All samples pass the Anderson-Darling test at the 5\% level. As before, switching to QMC drastically lowers the Anderson-Darling values, especially for AR QMC whose values are around $10^{-4}$. The Inverse QMC values range between 1.30 and $10^{-3}$.
\end{enumerate}

\subsection{Variance gamma model for option pricing}
The variance gamma (VG) model is a generalization of the classical Black-Scholes model for the dynamics of stock prices. The VG process $X(t;\sigma,\nu,\theta)$ is defined as
$$X(t;\sigma,\nu,\theta)=B(G(t;1,\nu);\theta,\sigma)$$
where $B(t;\theta,\sigma)$ is a Brownian motion and $G(t;1,\nu)$ is a Gamma process with a unit mean rate. The VG process, in other words, is a Brownian motion evaluated at a time given by a Gamma process.

The VG process can be simulated by sequential sampling, and bridge sampling. Within the sequential sampling approach, there are several algorithms to generate the process as well. A review of these generation algorithms can be found in \cite{fu}. These various algorithms require generation from the normal, gamma, and beta distributions. In our numerical results that follow, we used sequential sampling and the Gamma time-changed Brownian motion algorithm to generate the VG process.

In Table \ref{VG}, we report the price of a European call option with various maturities (given in the first column), when the underlying process is VG. The parameters of the option are taken from an example in \cite{webber}. The generation of the VG model with these particular parameters require the generation of the following distributions: ${\cal G}(0.83,0.3), {\cal G}(1.67,0.3), {\cal G}(2.5,0.3), {\cal G}(3.33,0.3)$ corresponding to $T=0.25, 0.5,0.75, 1.0$, as well as the normal distribution. We use the inverse transformation method to generate the normal distribution, and Algorithms 8 and 9 to generate the gamma distribution. We call this approach \textbf{AR MC} and \textbf{AR QMC}, for the MC and QMC implementations of these algorithms, in Table \ref{VG}. Methods \textbf{Inverse MC} and \textbf{Inverse QMC} generate gamma and normal distributions using the inverse transformation method.

In the \textbf{AR QMC} method, there are two cases. If the maturity is $T=0.25$, then we generate a 4-dimensional (randomized) QMC vector $(q_1, q_2, q_3, q_4)$. The first component is used to sample from the normal distribution using the inverse transformation method, and the last three components are used to sample from the gamma distribution using Algorithm 9. If $T>0.25$, then we generate a 3-dimensional (randomized) QMC vector, use its first component to sample from the normal distribution by the inverse transformation method, and use the last two components to sample from the gamma distribution by Algorithm 8. In the \textbf{Inverse QMC} method, to obtain one option price, we generate a 2-dimensional (randomized) QMC vector $(q_1, q_2)$. The first component is used to sample from the normal distribution, and the second component is used to sample from the gamma distribution.

For each maturity, we compute the option price by generating 10,000 stock price paths. We then independently repeat this procedure 100 times. The sample standard deviation of the resulting 100 estimates is reported in Table \ref{VG} (Std dev), together with the average of the 100 estimates (Price), and the total computing time. The last column reports the ratio of the efficiency of AR QMC to Inverse QMC. Here efficiency is the product of sample variance and computing time.  The second column ``Exact price" reports the analytical value of the option price (see \cite{madan2}). These values are taken from an example in \cite{webber}.

The Inverse QMC and AR QMC methods give estimates that agree with the exact solution to the cent. However, the acceptance-rejection algorithms run faster than the inverse algorithms by factors between 13 and 19. The AR QMC method has better efficiency than Inverse QMC, and the factors of improvement are between 2 and 6.

\begin{table}[h]

\centering
\begin{tabular}{ccc|ccccc}
\hline
\multirow{2}{*} {Maturity} & Exact && Inverse&AR &Inverse&AR&QMC EffGain\\
 & price & & MC & MC & QMC & QMC & AR/Inv\\
\hline
\hline

\multirow{3}{*}{0.25}&&Std dev&3.9e-2&3.5e-2&2e-3&3e-3&\\
&3.47&Price&3.47&3.47&3.47&3.47&6\\

&&Time(s)&9.55&0.74&9.56&0.71\\
\hline

\multirow{3}{*}{0.50}&&Std dev&6.2e-2&6.3e-2&2e-3&5e-3&\\
&6.24&Price&6.23&6.25&6.24&6.24&3\\

&&Time(s)&10.43&0.66&10.9&0.64&\\
\hline

\multirow{3}{*}{0.75}&&Std dev&7.9e-2&9.3e-2&3e-3&7e-3&\\
&8.69&Price&8.68&8.70&8.69&8.69&3\\

&&Time(s)&11.18&0.65&11.29&0.62&\\
\hline

\multirow{3}{*}{1.00}&&Std dev&1.06e-1&1.03e-1&3e-3&1e-2&\\
&10.98&Price&10.99&10.99&10.98&10.98&2\\

&&Time(s)&11.86&0.63&11.86&0.61&\\
\hline
\end{tabular}
\caption{Comparison of inverse and acceptance-rejection methods in pricing European call options in the variance gamma model. The option parameters are: $\theta=-0.1436$, $\sigma=0.12136$, $\nu=0.3$, initial stock price $S_0=100$, strike price $K=101$, and risk free interest rate $r=0.1$.}
\label{VG}
\end{table}

\section{Conclusions}

The use of low-discrepancy sequences in computational problems, especially in numerical integration, is increasing mainly because of the faster convergence rates these sequences provide, compared to pseudorandom sequences. For example, in the application of derivative pricing from computational finance, this faster rate of convergence is quite useful, and some well known low-discrepancy sequences have taken their place in the numerical methods toolbox of financial engineers.

Currently, the main method for transforming low-discrepancy sequences to nonuniform distributions is the inverse transformation technique. However, this technique can be computationally expensive for complicated distributions. The acceptance-rejection technique was developed precisely for this reason for pseudorandom sequences. In this paper, we presented theoretical and numerical results to argue that the acceptance-rejection technique is similarly useful in the context of low-discrepancy sequences. The availability of acceptance-rejection for low-discrepancy sequences significantly increases the scope of applications where quasi-Monte Carlo methods can improve traditional Monte Carlo. There is an extensive literature on efficient Monte Carlo algorithms for generating distributions, and many of them are based on acceptance-rejection. The results of this paper motivate the study of quasi-Monte Carlo versions of these algorithms.

\section{Acknowledgement}
We thank Dr. Burkardt, Department of Scientific Computing, Florida State University, for the inverse transformation codes used in this paper.

\end{document}